\documentclass[11pt,a4paper,DIV12]{scrartcl}

\usepackage{MCTSP}

\title{On Approximating Multi-Criteria TSP\thanks{An extended abstract of this
   work has appeared in the \emph{Proceedings of the 26th International
   Symposium on Theoretical Aspects of Computer Science
   (STACS 2009)}~\cite{Manthey:MCMinMax:2009}.}}

\author{Bodo Manthey}

\date{\small University of Twente, Department of Applied Mathematics \\
      P.~O.~Box 217, 7500 AE Enschede, The Netherlands \\
      \texttt{b.manthey@utwente.nl}}

\begin{document}

%%%%%%%%%%%%%%%%%%%%%%%%%%%%%%%%%%%%%%%%%%%%%%%%%%%%%%%%%%%%%%%%%%%%%%%%%%%%%%%%

\maketitle

%%%%%%%%%%%%%%%%%%%%%%%%%%%%%%%%%%%%%%%%%%%%%%%%%%%%%%%%%%%%%%%%%%%%%%%%%%%%%%%%

\begin{abstract}
  We present approximation algorithms for almost all variants of the
  multi-criteria traveling salesman problem (TSP).

  First, we devise randomized approximation algorithms for multi-criteria
  maximum traveling salesman problems (\maxtsp). For multi-criteria \maxstsp,
  where the edge weights have to be symmetric, we devise an algorithm with an
  approximation ratio of $2/3 - \eps$. For multi-criteria \maxatsp, where the
  edge weights may be asymmetric, we present an algorithm with a ratio of
  $1/2 - \eps$. Our algorithms work for any fixed number $k$ of objectives.
  Furthermore, we present a deterministic algorithm for bi-criteria \maxstsp\
  that achieves an approximation ratio of $7/27$.

  Finally, we present a randomized approximation algorithm for the asymmetric
  multi-criteria minimum TSP with triangle inequality (\minatsp). This algorithm
  achieves a ratio of $\log n + \eps$.
\end{abstract}

%%%%%%%%%%%%%%%%%%%%%%%%%%%%%%%%%%%%%%%%%%%%%%%%%%%%%%%%%%%%%%%%%%%%%%%%%%%%%%%%

%%%%%%%%%%%%%%%%%%%%%%%%%%%%%%%%%%%%%%%%%%%%%%%%%%%
\section{Multi-Criteria Traveling Salesman Problem}
\label{sec:mctsp}
%%%%%%%%%%%%%%%%%%%%%%%%%%%%%%%%%%%%%%%%%%%%%%%%%%%

%%%%%%%%%%%%%%%%%%%%%%%%%%%%%%%%%%%%%%%
\subsection{Traveling Salesman Problem}
\label{ssec:atsp}

The traveling salesman problem (TSP) is one of the most famous combinatorial
optimization problems. Given a graph, the goal is to find a Hamiltonian cycle of
maximum or minimum weight (\maxtsp\ or \mintsp). An instance of \emph{\maxtsp}
is a complete graph $G=(V,E)$ with edge weights $w: E \to \ratio_+$. The goal is
to find a Hamiltonian cycle of maximum weight. The weight of a Hamiltonian cycle
(more general, of any set of edges) is the sum of the weights of its edges. If
$G$ is undirected, we have \emph{\maxstsp} (symmetric TSP). If $G$ is directed,
we obtain \emph{\maxatsp} (asymmetric TSP).

An instance of \emph{\mintsp} is also a complete graph $G$ with edge weights $w$
that fulfill the triangle inequality: $w(u,v) \leq w(u,x) + w(x,v)$ for all
$u,v,x \in V$. The goal is to find a Hamiltonian cycle of minimum weight. We
have \emph{\minstsp} if $G$ is undirected and \emph{\minatsp} if $G$ is
directed.

All these variants are NP-hard and APX-hard. Thus, we have to content ourselves
with approximate solutions. The currently best approximation algorithm for
\maxstsp\ achieves an approximation ratio of
$7/9$~\cite{PaluchEA:MaxTSP79:2009}, and the currently best algorithm for
\maxatsp\ achieves a ratio of $2/3$~\cite{KaplanEA:TSP:2005}. \minatsp\ can be
approximated with a factor of $O(\log n/\log\log n)$, where $n$ is the number of
vertices of the instance~\cite{AsadpourEA:ATSPolog:2010}.

Cycle covers are often used for designing approximation algorithms for the
TSP~\cite{BlaeserManthey:MWCC:2005,KaplanEA:TSP:2005,FeigeSingh:ATSP:2007}.
A \emph{cycle cover} of a graph is a set of vertex-disjoint cycles such that
every vertex is part of exactly one cycle. The general idea is to compute an
initial cycle cover and then to join the cycles to obtain a Hamiltonian cycle.
This technique is called
\emph{subtour patching}~\cite{GilmoreEA:WellSolved:1985}. Hamiltonian cycles are
special cases of cycle covers that consist of a single cycle. Thus, the weight
of a maximum-weight cycle cover bounds the weight of a maximum-weight
Hamiltonian cycle from above, and minimum-weight cycle covers provide lower
bounds for minimum-weight Hamiltonian cycles. In contrast to Hamiltonian cycles,
cycle covers of maximum or minimum weight can be computed efficiently by
reduction to matching problems~\cite{AhujaEA:NetworkFlows:1993}.

%%%%%%%%%%%%%%%%%%%%%%%%%%%%%%%%%%%%%%%%
\subsection{Multi-Criteria Optimization}
\label{ssec:mcopt}

In many optimization problems, there is more than one objective function. This
is also the case for the TSP: We might want to minimize travel time, expenses,
number of flight changes, etc., while maximizing, e.g., our profit
along the way. This leads to $k$-criteria variants of the TSP (\emph{\kcmaxstsp k},
\emph{\kcmaxatsp k}, \emph{\kcminstsp k}, and \emph{\kcminatsp k}
for short).

With respect to a single criterion, the term ``optimal solution'' is
well-defined. However, if several criteria are involved, there is no natural
notion of a best choice. Instead, we have to be satisfied with trade-off
solutions. The goal of multi-criteria optimization is to cope with this dilemma.
To transfer the concept of optimal solutions to multi-criteria optimization
problems, the notion of \emph{Pareto curves} (also known as \emph{Pareto sets}
or \emph{efficient sets}) has been introduced (cf.\
Ehrgott~\cite{Ehrgott:MulticriteriaOpt:2005}). A Pareto curve is a set of
solutions that can be considered optimal.

We introduce the following terms only for maximization problems. After that, we
briefly state the differences for minimization problems. An instance of
\kcmaxtsp k is a complete graph $G$ with edge weights
$w_1, \ldots, w_k: E \to \ratio_+$. A Hamiltonian cycle $H$ \emph{dominates} another
Hamiltonian cycle $\tilde H$ if $w_i(H) \geq w_i(\tilde H)$ for all $i \in [k] = \{1, \ldots, k\}$ and
$w_i(H) > w_i(\tilde H)$ for at least one $i$. This means that $H$ is strictly
preferable to $\tilde H$. A \emph{Pareto curve} of solutions contains all solutions
that are not dominated by another solution. For other maximization problems,
$k$-criteria variants are defined analogously.

Unfortunately, Pareto curves cannot be computed efficiently in many cases:
First, they are often of exponential size. Second, because of
reductions from knapsack problems, they are NP-hard to compute even for
otherwise easy optimization problems. Third, TSP is NP-hard already with
only one objective function, and optimization problems do not become easier with
more objectives involved. Therefore, we have to be satisfied with approximate
Pareto curves.

For simpler notation, let $w(H) = (w_1(H), \ldots, w_k(H))$.
Inequalities are meant component-wise. A set $\mathcal P$ of Hamiltonian
cycles of $V$ is called an \emph{$\alpha$ approximate Pareto curve} for $(G, w)$
if the following holds: For every Hamiltonian cycle $\tilde H$, there exists a
Hamiltonian cycle $H \in \mathcal P$ with $w(H) \geq \alpha w(\tilde H)$. We have
$\alpha \leq 1$, and a $1$ approximate Pareto curve is a Pareto curve.

An algorithm is called an \emph{$\alpha$ approximation algorithm} if, given $G$
and $w$, it computes an $\alpha$ approximate Pareto curve. It is called a
\emph{randomized $\alpha$ approximation} if its success probability is at least
$1/2$. This success probability can be amplified to $1-2^{-m}$ by executing the
algorithm $m$ times and taking the union of all sets of solutions. (We can also
remove solutions from this union that are dominated by other solutions in the
union.) Again, the concepts can be transfered easily to other maximization
problems.

Papadimitriou and Yannakakis~\cite{PapadimitriouYannakakis:TradeOffs:2000}
have shown that $(1-\eps)$ approximate Pareto curves of size polynomial in the
instance size and $1/\eps$ exist. The technical requirement for the existence is
that the objective values of all solutions for an instance $X$ are either $0$
or within an interval $[2^{-p(N)}, 2^{p(N)}]$
for some polynomial $p$, where $N$ is the size of $X$. This
is fulfilled by most optimization problems and in particular in our case.

A \emph{fully polynomial time approximation scheme} (FPTAS) for a multi-criteria
optimization problem computes $(1-\eps)$ approximate Pareto curves in time
polynomial in the size of the instance and $1/\eps$ for all $\eps > 0$.
Multi-criteria maximum-weight matching admits a
\emph{randomized FPTAS}~\cite{PapadimitriouYannakakis:TradeOffs:2000}, i.~e.,
the algorithm succeeds in computing a $(1-\eps)$ approximate Pareto curve with
a probability of at least $1/2$. This randomized FPTAS yields also a
randomized FPTAS for
the multi-criteria maximum-weight cycle cover
problem~\cite{MantheyRam:MultiCritTSP:2009}.

To define Pareto curves and approximate Pareto curves also for minimization
problems, in particular for \kcminstsp k\ and \kcminatsp k, we have to replace all
``$\geq$'' and ``$>$'' above by ``$\leq$'' and ``$<$''. Furthermore, $\alpha$
approximate Pareto curves are now defined for $\alpha \geq 1$, and an FPTAS achieves
an approximation ratio of $1+\eps$. There also exists a randomized
FPTAS for the multi-criteria minimum-weight cycle cover problem.

%%%%%%%%%%%%%%%%%%%%%%%%%
\subsection{Known Results}
\label{ssec:related}

For an overview of the literature about multi-criteria optimization, including
multi-criteria TSP, we refer to Ehrgott and
Gandibleux~\cite{EhrgottGandibleux:Multiobjective:2000,Ehrgott:MulticriteriaOpt:2005}.

Angel et al.~\cite{AngelEA:BicritTSP:2004,AngelEA:MultiTSP:2005} have considered
\kcminstsp k\ restricted to edge weights $1$ and $2$. They analyzed a local search
heuristic and proved that it achieves an approximation ratio of $3/2$ for $k = 2$
and of $2 - \frac{2}{k+1}$ for $k \geq 3$. Ehrgott~\cite{Ehrgott:MultiTSP:2000}
has analyzed a variant of \kcminstsp k, where all objectives are encoded into a
single objective by using some norm. He proved approximation ratios between
$3/2$ and $2$ for this problem, where the ratio depends on the norm used.
\kcminstsp k\ can be approximated with a ratio of
$2+\eps$~\cite{MantheyRam:MultiCritTSP:2009}.
For \kcminatsp k, we are not aware of any approximation algorithm.

Bl\"aser et al.~\cite{BlaeserEA:MCMaxTSP:2008} have devised the first randomized
approximations for \kcmaxstsp k\ and \kcmaxatsp k. Their algorithms achieve ratios
of $\frac 1k - \eps$ for \kcmaxstsp k and $\frac 1{k+1} - \eps$ for
\kcmaxatsp k. They have conjectured that
approximation ratios of $\Omega(1/\log k)$ are possible.

%%%%%%%%%%%%%%%%%%%%%%%%
\subsection{New Results}
\label{ssec:new}

We devise approximation algorithms for \kcmaxstsp k, \kcmaxatsp k, and \kcminatsp k\
that work for any number $k$ of criteria.

First, we solve the conjecture of Bl\"aser et al.~\cite{BlaeserEA:MCMaxTSP:2008}
affirmatively. We even prove a stronger result:
For \kcmaxstsp k, we achieve a ratio
of $2/3 - \eps$, while for \kcmaxatsp k, we achieve a ratio of $1/2 - \eps$
(Section~\ref{sec:alg}).
The general idea of these algorithms is sketched in
Section~\ref{sec:idea}. After that, we introduce a decomposition technique in
Section~\ref{sec:decomp} that will lead to our algorithms (Section~\ref{sec:alg}).
Our algorithms are randomized and their running-time is polynomial in the input size for any fixed
$\eps > 0$ and any fixed number $k$ of criteria.

Furthermore, as a first step towards deterministic approximation algorithms for
\kcmaxtsp k, we devise an approximation algorithm for \kcmaxstsp 2 that achieves
an approximation ratio of $7/27$ (Section~\ref{sec:detbi}).
As a side effect, this result proves that for
\kcmaxstsp 2, there always exists a single Hamiltonian cycle that already is a
$1/3$ approximate Pareto curve. This does not
hold for any other variant of multi-criteria TSP.

Finally, we devise the first approximation algorithm for \kcminatsp k\
(Section~\ref{sec:atsp}).
The approximation
ratio of our algorithm is $\log n + \eps$, where $n$ is the
number of vertices.
Our algorithm is randomized and its running-time is polynomial in the input size
and in $1/\eps$ for any fixed number of criteria.

%%%%%%%%%%%%%%%%%%%%%%%%%%%%%%%%%%%%%%%%%
\section{Idea for Multi-Criteria Max-TSP}
\label{sec:idea}
%%%%%%%%%%%%%%%%%%%%%%%%%%%%%%%%%%%%%%%%%

For \maxatsp, we can easily get a $1/2$ approximation: We compute a
maximum-weight cycle cover and remove the lightest edge of each cycle. This yields
a collection of paths. Then we add edges to connect the
paths, which yields a Hamiltonian cycle. For \maxstsp, this approach gives a
ratio of $2/3$ since the length of every cycle is at least three.

Unfortunately, this does not generalize to multi-criteria \maxtsp. The reason is
that the term ``lightest edge'' is not
well defined: An edge that has little weight with respect to one objective might
have a huge weight with respect to another objective. Based on this
observation, the basic idea behind our algorithms is ``guessing'' the heavy edges
such that the remaining edges are all light-weight. A similar technique has already been
used by Ravi and Goemans~\cite{RaviGoemans:ConstrainedSpanningTree:1996} for bi-criteria
spanning trees.
Since the remaining edges are light-weight, and we can 
break one edge of every cycle without losing two much weight with respect to any objective function.
This is based on the decompositions introduced in the following section.

%%%%%%%%%%%%%%%%%%%%%%%%%%%%%%%%%
\section{Improved Decompositions}
\label{sec:decomp}
%%%%%%%%%%%%%%%%%%%%%%%%%%%%%%%%%

Given a cycle cover $C$, a decomposition of $C$ is a collection $P \subseteq C$
of paths. From such a collection $P$, we obtain a Hamiltonian cycle just by
connecting the endpoints of the paths appropriately.
In order to get approximation algorithms for multi-criteria Max-TSP, our goal is to find
collections $P$ with $w(P) \geq \alpha w(C)$ for an $\alpha$ as large as possible.
Decompositions have already been used by Bl\"aser et al.~\cite{BlaeserEA:MCMaxTSP:2008}
for their approximation algorithms for multi-criteria Max-TSP. With their decompositions,
they have achieved approximation ratios of $\frac{1}{k} - \eps$ and $\frac{1}{k+1} - \eps$,
and they conjectured that approximation ratio $\Omega(1/\log k)$ is possible.
We introduce a slightly different kind of decompositions, which enables us to design
constant-factor approximations.

Let $C$ be a cycle cover, and let $w=(w_1, \ldots, w_k)$ be edge weights. We say
that the pair $(C, w)$ is $\eta$-light for some $\eta\leq 1$ if
$w(e) \leq \eta w(C)$ for all $e \in C$.
From now on, let $\eta_{k,\eps} = \frac{\eps^2}{2\ln k}$.

\begin{theorem}
\label{thm:improvedec}
  Let $\eps \in (0,\frac 12)$ be arbitrary, and let $k \geq 2$ be
  arbitrary. Let $C$ be a cycle cover, and let $w = (w_1, \ldots, w_k)$ be edge
  weights such that $(C, w)$ is $\eta_{k,\eps}$-light.

  If $C$ is directed, then there exists a collection $P \subseteq C$ of paths
  with $w(P) \geq (\frac 12 - \eps) \cdot w(C)$.

  If $C$ is undirected, then there exists a collection $P \subseteq C$ of paths
  with $w(P) \geq (\frac 23 - \eps) \cdot w(C)$.
\end{theorem}

\begin{proof}

  The proof uses Hoeffding's inequality~\cite[Theorem 2]{Hoeffding:SumsBounded:1963}.

\begin{lemma}[Hoeffding's inequality]
\label{lem:hoeffding}
  Let $X_1, \ldots, X_n$ be independent random variables, where $X_j$ assumes
  values in the interval $[a_j, b_j]$. Let $X = \sum_{j = 1}^n X_j$. Then
  \[
         \probab\bigl(X < \expected(X) - t\bigr)
    \leq \exp\left(-\frac{2 t^2}{\sum_{j = 1}^n (b_j-a_j)^2}\right).
  \]
\end{lemma}

  We start by considering the directed
  case. Let $C$ be a directed cycle cover with edge weights $w$ such that
  $(C,w)$ is $\eta_{k, \eps}$-light. We scale the weights
  such that $w(C) = 1/\eta_{k, \eps}$. Thus, $w(e) \leq 1$ for all $e \in C$.

  Let $c_1,\ldots, c_m$ be the cycles of $C$, and consider any cycle $c_j$ of $C$.
  We choose one edge of $c_j$ for removal uniformly at random. By doing this for
  $j \in [m]$, we obtain a decomposition $P$ of $C$. Fix any objective $i$. Let
  $X_j = \sum_{e \in c_j \cap P} w_i(e)$ be the random variable of the
  contribution of $c_j$ to the weight $w_i(P)$. Since $w_i(e) \in [0,1]$ for all
  $e \in C$, there exist $a_j, b_j \in \real$ such that $X_j$ assumes only
  values in $[a_j, b_j]$ and $0 \leq b_j - a_j \leq 1$. Let
  $X=\sum_{j=1}^m X_j = w_i(P)$ be the random variable of the weight of $P$ with
  respect to objective $i$. Since every cycle has a length of at least two,
  the probability of deleting any fixed edge is at most $1/2$.
  Thus, by linearity of expectation, we have $\expected(X)  \geq A/2$.

  If we can show that $\probab(X < (\frac 12 - \eps) \cdot A) < 1/k$, then, by a union bound,
  $\probab(\exists i \in [k]: w_i(P) < (\frac 12 - \eps) \cdot A) < 1$, which would imply the existence
  of a decomposition $P$ as claimed.
  Since $0 \leq b_j - a_j \leq 1$, we have
  \[
         w_i(C) = A  = \sum_{j=1}^m w_i(c_j)
    \geq \sum_{j=1}^m b_j
    \geq \sum_{j=1}^m b_j - a_j
    \geq \sum_{j=1}^m (b_j - a_j)^2.
  \]
  Plugging this into Hoeffding's bound yields
  \[
         \probab\left(w_i(P) < \left(\frac 12 - \eps\right) \cdot w_i(C)\right)
    \leq \exp\left(- \frac{2 \eps^2 w_i(C)^2}{w_i(C)}\right)
 <  \frac 1{k^2}
  \]
  for $k \geq 2$ since $w_i(C) = 2\ln k/\eps^2 = 1/\eta_{k, \eps}$.
  The proof for undirected cycle covers is identical except for $\expected(X) \geq 2A/3$ and
  therefore omitted.
\end{proof}

We now know that decompositions exist. But, in
order to use them in approximation algorithms, we have to find them efficiently.
Theorem~\ref{thm:improvedec} immediately yields a randomized algorithm:
Assume that we have an $\eta_{k, \eps}$-light
pair $(C,w)$. We randomly select one
edge of every cycle of $C$ for removal and put all remaining edges
into $P$. The probability that $P$ is not a $(\frac 12 - \eps)$-
or $(\frac 23 - \eps)$-decomposition (depending on whether $C$
is directed or undirected) is bounded from above by $1/k \leq 1/2$. Thus, we
obtain a feasible decomposition with constant probability. We iterate
this process until a feasible decomposition is found.

For the deterministic algorithm, we assume again that we have an
$\eta_{k,\eps}$-light pair $(C,w)$. We scale the weights such that
$w_i(C) = 1/\eta_{k,\eps}$ for all $i$. Then $w(e) \leq 1$ for all $e \in C$.
The main idea is to reduce an arbitrary instance to a
new instance whose size depends only on $k$ and $\eps$.

First, we normalize our cycle cover such that they consist solely of cycles of
the shortest possible length. For directed cycle covers $C$, we can restrict
ourselves to cycles of length two: Any cycle $c$ of length $\ell$ with edges
$e_1, \ldots, e_\ell$ can be replaced by $\lfloor \ell/2 \rfloor$ cycles
$(e_{2j-1}, e_{2j})$ for $j = 1, \ldots, \lfloor \ell/2 \rfloor$. If $\ell$ is
odd, then we add an edge $e_{\ell+1}$ with $w(e_{\ell+1}) = 0$ and add the cycle
$(e_\ell, e_{\ell+1})$. (Technically, edges consist of vertices, and we cannot
simply reconnect them. What we mean is that we create new cycles of length two,
and the edges of those cycles have the same names and
the same weights as in the original cycles.)
We do this for all cycles of length at least three and
call the resulting cycle cover $C'$. Now any decomposition $P'$ of $C'$ yields a
decomposition $P$ of the original cycle cover $C$ by removing the newly added
edges $e_{\ell+1}$ if they are in $P'$. Furthermore, $w_i(e) \leq 1$ for the new
cycle cover $C'$. Analogously, undirected cycle covers can be normalized to
consist solely of cycles of length three.

Second, assume that we have two cycles $c$ and $c'$ in a normalized cycle cover
with $w(c) + w(c') \leq 1$. Then we can combine $c$ and $c'$ to $\tilde c$: Let
$e_1,e_2$ and $e'_1, e'_2$ be the edges of $c$ and $c'$, respectively.
Then we can replace $e_i$ and $e'_i$
by $\tilde e_i$ with $w(\tilde e_i) = w(e_i) + w(e'_i)$. The cycle cover plus
edge weights thus obtained are still $\eta_{k,\eps}$-light. We continue
combining cycles greedily until no more combinations are possible.
For undirected cycles, we proceed analogously. The difference is that the cycles
consist of three edges.
The resulting cycle cover contains at most $2k/\eta_{k,\eps}$ cycles.
Thus, an optimal decomposition can be found with a running-time that now only
depends on $k$ and $\eps$.

Overall, for every fixed $\eps > 0$ and $k \geq 2$, we have a deterministic
algorithm that, given an $\eta_{k, \eps}$-light directed cycle cover $C$ with
edge weights $w$, computes a 
$(\frac 12 - \eps)$-decomposition of $C$ in time polynomial in the input size.
If $C$ is undirected, a
$(\frac 23 - \eps)$-decomposition can be computed analogously.
We call these algorithms \lw\ with parameters $C$, $w$, and~$\eps$.

%%%%%%%%%%%%%%%%%%%%%%%%%%%%%%%%%%%%%%%%%%%%%%%%%%%%%%%%%%%%%
\section{Approximation Algorithms for Multi-Criteria Max-TSP}
\label{sec:alg}
%%%%%%%%%%%%%%%%%%%%%%%%%%%%%%%%%%%%%%%%%%%%%%%%%%%%%%%%%%%%%

In this section, $\ca$ denotes the randomized FPTAS for cycle covers. More
precisely, let $G$ be a graph (directed or undirected), $w=(w_1, \ldots, w_k)$
be edge weights, $\eps > 0$, and $p > 0$. Then
$\algcc(G,w, \eps, p)$ yields a $(1-\eps)$-approximate Pareto curve
of cycle covers of $G$ with weights $w$ with a success probability of at least
$1-p$.

%%%%%%%%%%%%%%%%%%%%%%%%%%%%%%%%%%%%
\subsection{Multi-Criteria \maxatsp}
\label{ssec:algatsp}

Our goal is to guess small sets of heavy edges such that
decomposition on the remaining graph is possible.
To do so, we need the following notation.
For a graph $G=(V,E)$ and a subset $K \subseteq E$ of $G$'s edges
($K$ forms a subset of a Hamiltonian cycle),
we get $G_{-K}$ by contracting all edges of $K$. Contracting a single
edge $(u,v)$ means removing all outgoing edges of $u$, removing
all incoming edges of $v$, and identifying $u$ and $v$.
Analogously, for a Hamiltonian cycle $H$ and edges $K$, we obtain
a Hamiltonian cycle $H_{-K}$ of $G_{-K}$ by contracting the edges in $K$.
If $(G, w)$ is an instance, then
$(G_{-K}, w)$ denotes the instance with $w$ modified according to the
edge contractions.

We now prove the following: For every Hamiltonian cycle $\tilde H$, there exists a (small) set
$K \subseteq \tilde H$ of edges such that $\tilde H_{-K}$ is light-weight.
Therefore, let
\[
    f(k,\eps)
  =       k
    \cdot \left\lceil
            \frac{\log(1/2 + \eps)}
                 {\log\bigl(1-\eta_{k,\eps/2} + (\frac \eps2)^3\bigr)}
          \right\rceil.
\]
From now on, we assume that $\eps \in (0, \frac 1{2\ln k})$ is fixed.

\begin{lemma}
\label{lem:dirsubset}
  For every Hamiltonian cycle $\tilde H$ and every $\eps > 0$, there exists a
  subset $K \subseteq \tilde H$ such
  that $|K| \leq f(k,\eps)$ and, for every $i \in [k]$, we have
  \begin{enumerate}
  \item $w_i(K) \geq (\frac 12 - \eps) \cdot w_i(\tilde H)$ or
     \label{enough}
  \item $w_i(e) \leq \bigl(\eta_{k, \eps/2} - (\frac \eps2)^3\bigr) \cdot w_i(\tilde H_{-K})$
     for all $e \in \tilde H_{-K}$.
     \label{light} 
  \end{enumerate}
\end{lemma}

\begin{proof}
We put edges one by one into $K$ until the properties are met for all objectives.
If not all $i$ fulfill Property~\ref{enough} or~\ref{light} yet, then we have to
add another vertex to $K$.
Let $\xi = \eta_{k, \eps/2} - (\frac \eps2)^3$ for short.
There exists an edge $e \in \tilde H \setminus K$
and an $i \in [k]$ such that
$w_i(e) > \xi w_i(\tilde H_{-K})$
and $w_i(K) < (\frac 12 - \eps) \cdot w_i(\tilde H)$.
We say that this $i$ is the winner of round $j$,
and we add $e$ to $K$. Let us call the new set
$K' = K \cup \{e\}$.

Whenever an $i \in [k]$ is a winner, we have
\[
  \frac{w_i(\tilde H_{-K'})}{w_i(\tilde H_{-K})} \leq 1-\xi \: .
\]
By definition, we have $w(K) + w(\tilde H_{-K}) = w(\tilde H)$.
Thus, if $i$ has won $\ell$ rounds on the way to $K$, we can conclude
that $w_i(K) \geq \bigl(1-(1-\xi)^\ell \bigr) \cdot w_i(\tilde H)$.
For $\ell =\bigl\lceil \frac{\log 1/2 + \eps}{\log(1-\xi)}\bigr\rceil$,
we have
$w_i(K) \geq (\frac 12 - \eps) \cdot w_i(\tilde H)$.
Observing that every round has a winner completes the proof.
\end{proof}

Now we know that few edges for $K$ suffice to make $\tilde H_{-K}$
light-weight.
But given the set $K$, how do we find an appropriate cycle cover?
Therefore, let $\beta_i = \max\{w_i(e) \mid e \in \tilde H_{-K}\}$ be the
weight of the heaviest edge with respect to the $i$th objective. Let
$\beta = \beta(\tilde H_{-K}) = (\beta_1,\ldots, \beta_k)$.
We modify our edge weights $w$ to $w^\beta$ as follows:
\[
  w^\beta(e) = \begin{cases}
                 w(e) & \text{ if $w(e) \leq \beta$ and} \\
                 0    & \text{ if $w_i(e) > \beta_i$ for some $i$.}
                 \end{cases}
\]
This means that we set all edge weights exceeding $\beta$ to $0$. Since
$\tilde H_{-K}$ does not contain any of those edges, we have
$w(\tilde H_{-K}) = w^\beta(\tilde H_{-K})$.
The advantage of $w^\beta$ is that, if we compute a
$(1-\eps)$ approximate Pareto curve $\pc^\beta$ of cycle covers with edge
weights $w^\beta$, we obtain a cycle cover to which we can apply decomposition
to obtain a collection $P$ of paths.
This is stated in the following lemma.

\begin{lemma}
\label{lem:betapower}
  Let $\nu > 0$ be arbitrary. Let $H_0$ be a directed Hamiltonian cycle
  with $w(e) \leq \bigl(\eta_{k,\nu}  - \nu^3\bigr) \cdot w(H_0)$ for all
  $e \in H_0$. Let $\beta = \beta(H_0)$, and let $\pc$ be a
  $(1-\nu)$ approximate Pareto curve of cycle covers with respect to~$w^\beta$.

  Then $\pc$ contains a cycle cover $C$ with
  $w^\beta(C) \geq (1-\nu) \cdot w(H_0)$ and
  $w^\beta(e) \leq \eta_{k,\nu} \cdot w^\beta(C)$ for all $e \in C$. This
  cycle cover $C$ yields a decomposition $P \subseteq C$ with
  $w(P) \geq (\frac 12 - 2\nu) \cdot w(H_0)$.
\end{lemma}

\begin{proof}
  Since the Hamiltonian cycle $H_0$ is in particular a cycle cover, the set
  $\pc$ contains a cycle cover $C$ with $w^\beta(C) \geq
  (1-\nu) \cdot w^\beta(H_0) = (1-\nu) \cdot w(H_0)$. For every
  edge $e \in C$ and every $i$, we have $w_i^\beta(e) \leq
  (\eta_{k,\nu} - \nu^3) \cdot w^\beta_i(H_0) \leq
  \frac{\eta_{k,\nu} - \nu^3}{1-\nu} \cdot w_i^\beta(C) \leq
  \eta_{k,\nu} \cdot w_i(C)$. The last inequality follows from
  $\eta_{k,\nu} - \nu^3 \leq \eta_{k,\nu} \cdot (1-\nu)$. This is
  equivalent to $\nu^2 \geq \eta_{k,\nu} = \frac{\nu^2}{2\ln k}$, which is valid.

  The cycle cover $C$ can be decomposed into a collection $P \subseteq C$ of
  paths with
  $w(P) \geq w^\beta(P) \geq (\frac 12 - \nu) \cdot w^\beta(C) \geq (\frac 12 - \nu) \cdot (1-\nu) \cdot w(H_0)
  \geq (\frac 12 - 2\nu) \cdot w(H_0)$ by Theorem~\ref{thm:improvedec}.
\end{proof}

For $H_0 = \tilde H_{-K}$, the set $P$ approximates $\tilde H_{-K}$, and $P \cup K$ yields
a tour $H$ that approximates $\tilde H$.
However, to obtain an algorithm, we have to find $\beta$ and $K$.  So far, we have assumed
that we already know the Hamiltonian cycles that we aim for. But there is only a
polynomial number of possibilities for $\beta$ and $K$: For all $\beta$ and
for all $i \in [k]$, we can
assume that there is an edge with $w_i(e) = \beta_i$. Thus,
there are at most $O(n^2)$ choices for $\beta_i$, hence at most
$O(n^{2k})$ in total. The cardinality of $K$ is bounded in terms of $f(k, \eps)$
as shown in the lemma above. For fixed $k$ and $\eps$, there is only a
polynomial number of subsets of cardinality at most $f(k,\eps)$.

Overall, we obtain \apa\ (Algorithm~\ref{alg:maxatsp}) and the following
theorem.

\begin{algorithm}[t]
\begin{algorithmic}[1]
\item[] $\ptsp \leftarrow \atspalg(G, w, \eps)$
\Input directed complete graph $G=(V,E)$, $w: E \to \ratio_+^k$, $\eps > 0$
\Output $(\frac 12 - \eps)$ approximate Pareto curve \ptsp\ for
        \kcmaxatsp k with probability at least $\frac 12$
\ForAll{$K \subseteq E$ with $|K| \leq f(k, \eps)$ that form a subset of a tour and bounds
        $\beta$}
   \State       $\pc_{K,\beta}
      \leftarrow \algcc(G_{-K},w^\beta, \frac{\eps}2,
                             \frac{1}{2n^{2k+f(k,\eps)}})$
      \label{atsp:cyco}
   \ForAll{$I \subseteq [k]$ and $C \in \pc_{K, \beta}$}
      \If{$w^\beta_I(e) \leq \eta_{k,\eps/2} \cdot w^\beta_I(C)$ for all
           $e \in C$}
         \State $P \leftarrow
              \lightweight(C, w^\beta_I, \frac \eps2)$
         \State add edges to $P \cup K$ to obtain a Hamiltonian cycle $H$
         \State add $H$ to $\ptsp$
      \EndIf
   \EndFor
\EndFor
\end{algorithmic}
\caption{Approximation algorithm for \kcmaxatsp k.}
\label{alg:maxatsp}
\end{algorithm}

\begin{theorem}
  \label{thm:mcatsp}
  For every fixed $k \geq 2$ and $\eps > 0$, \apa\ is a randomized
  $\frac 12 - \eps$ approximation for $k$-criteria \maxatsp\ whose running-time
  is polynomial in the input size.
\end{theorem}

\begin{proof}
  Let us analyze the approximation ratio first. To do this, we assume that all
  randomized computations are successful. After that, we analyze success
  probability and running-time.

  Let $\tilde H$ be an arbitrary Hamiltonian cycle.
  We have to show that there exists a
  Hamiltonian cycle $H \in \ptsp$ with $w(H) \geq (\frac 12 - \eps) \cdot
  w(\tilde H)$.
  By Lemma~\ref{lem:dirsubset}, there exists a set $K \subseteq \tilde H$ of
  cardinality at most $f(k, \eps)$ and a
  set $I \subseteq [k]$ with the following properties:
  \begin{itemize}
  \item For every $i \in [k] \setminus I$, we have
     $w_i(K) \geq (\frac 12 - \eps) \cdot w_i(\tilde H)$.
  \item For every $i \in I$ and for every edge $e \in \tilde H_{-K}$, we have
     $w_i(e) \leq \bigl(\eta_{k, \eps/2} - (\frac \eps2)^3\bigr) \cdot w_i(\tilde H_{-K})$.
  \end{itemize}

  Let $\beta = \beta(\tilde H_{-K})$. According to Lemma~\ref{lem:betapower} with
  $H_0 = \tilde H_{-K}$ and $\nu = \eps/2$, the set
  $\pc_{K, \beta}$ contains a cycle cover $C$ that can be decomposed into
  a collection $P$ of paths such that
  $w_i(P) \geq (\frac 12 - \eps) \cdot w_i(\tilde H_{-K})$.
  The set $P \cup K$ is also a collection of paths. We get a Hamiltonian cycle $H \supseteq P \cup K$
  by adding arbitrary edges.
  For the weight of $H$, we have
  \[
    w_i(H) \geq w_i(K)
           \geq \left(\frac 12 - \eps\right) \cdot w(\tilde H)
  \]
  for every $i \in I$ and
  \[
    w_i(H) \geq w_i(P) + w_i(K)
           \geq \left(\frac 12 - \eps\right) \cdot w_i(\tilde H_{-K}) + w_i(K)
           \geq \left(\frac 12 - \eps\right) \cdot w_i(\tilde H)
  \]
  for every $i \in [k] \setminus I$ since $w_i(K) + w_i(\tilde H_{-K}) = w_i(\tilde H)$.
  This proves the approximation ratio.

  The running time and the error probability remain to be analyzed. The error
  probabilities of the randomized computations in line~\ref{atsp:cyco}
  are chosen such that they sum up to at most $1/2$. This
  yields that the probability that one of the computations fails is at most
  $1/2$. The running time follows since $f(k, \eps)$ is independent of $n$, the
  number of bounds $\beta$ is bounded by $n^{2k}$, and there are $2^k$ possible
  sets~$I$.
\end{proof}

%%%%%%%%%%%%%%%%%%%%%%%%%%%%%%%%%%%%
\subsection{Multi-Criteria \maxstsp}
\label{ssec:algstsp}

Our goal is again to show that, for any Hamiltonian cycle $\tilde H$, taking out a small set $K$ of heavy edges suffices
to make the rest of $\tilde H$ light-weight.
Unfortunately, contracting heavy edges in undirected graphs is not as easy as it is
in directed graphs: The statements ``remove all incoming'' and
``remove all outgoing'' edges are not well-defined in an undirected graph.

To circumvent these problems, we do not contract edges $e = \{u,v\}$. Instead, we
set the weight of all edges incident to $u$ or $v$ to $0$. This allows us to add
the edge $e$ to any collection $P$ of paths without decreasing the weight: We
remove all edges incident to $u$ or $v$ from $P$, and then we add $e$. The
result is again a collection of paths.

However, by setting the weight of edges adjacent to $u$ or $v$ to $0$, we might
destroy a lot of weight with respect to some objective. To circumvent this
problem as well, we put larger neighborhoods of the edges into $K$.
In this way, we can add our heavy-weight edge (plus some more edges of its
neighborhood) to the collection of paths without losing too much weight from
removing other edges. Lemma~\ref{lem:somelight} below justifies this.

\begin{lemma}
\label{lem:somelight}
  Let $\tilde H$ be a Hamiltonian cycle as described above, let
  $w=(w_1, \ldots, w_k)$ be edge weights,
  and let $e_1, \ldots, e_\ell$ be any $\ell$ distinct edges of $\tilde H$.
  Then there exists a $j \in [\ell]$ such that
  \[
    w(e_j) \leq \frac{k}{\ell} \cdot w(\tilde H).
  \]
\end{lemma}

\begin{proof}
  Suppose otherwise and assume without loss of generality that
  $w_i(\tilde H) > 0$ for all $i$. We scale the weights such that
  $w_i(\tilde H) = 1$ for all $i$. Then for all $j$ there is an $i_j$ with
  $w_{i_j}(e_j) > \frac{k}{\ell} \cdot w_{i_j}(\tilde H) = \frac k \ell$.
  Thus, $\sum_{j=1}^\ell \sum_{i=1}^k w_i(e_j) >
  \sum_{j=1}^\ell \frac{k}{\ell} \cdot w_{i_j}(\tilde H) = k$. But, since all
  edges are distinct, we also have $\sum_{j=1}^\ell \sum_{i=1}^k w_i(e_j) \leq
  \sum_{i=1}^k w_i(\tilde H) = k$ -- a contradiction.
\end{proof}

Let $\tilde H$ be an arbitrary Hamiltonian cycle. Let
$e_0, e_1, \ldots, e_{n-1}$ be the edges of $\tilde H$ in the order in which
they appear in $\tilde H$ ($e_0$ is chosen arbitrarily). Let
$e_j=\{v_j, v_{j+1}\}$, where arithmetic of the indices here and in the
following is modulo $n$.
Now let $e_0$ be a heavy-weight edge of $\tilde H$. Then we put $e_0$ into our set $K$,
and we set the
weight of all edges incident to $v_0$ and $v_1$ to $0$. But in this way, we
lose the weight of $e_1$ and $e_{-1}$. In order to maintain the
approximation ratio, we have to avoid that we lose too much weight. Therefore, we
consider paths that include $e_0$. If we set the weight of all edges incident to
$v_{p+1}, \ldots, v_{q}$ with $p < 0 < q$ to $0$, we lose
only the weight of $e_{p}$ and $e_{q}$. To keep track of things, we also
put $e_{p+1}, \ldots, e_{q-1}$ into $K$. Furthermore, we put the two
edges $e_{p}$ and $e_{q}$, whose weight might get lost, into a set $T$. By
Lemma~\ref{lem:somelight}, we can make sure that both $e_{p}$ and $e_{q}$ are
not too heavy. Finally, we put $v_{p+1}, \ldots, v_{q}$ into the set $L$, which is
the set of vertices whose adjacent edges have now weight $0$. We denote the
corresponding edge weights by $w^L$ (see below for a formal definition).

Given any collection of paths $P$, we can now remove all edges of weight $0$ and add
the edges $e_{p+1}, \ldots, e_{q-1}$ to obtain a new collection of paths.
This does not change the weight with respect to $w^L$. The only
edges that we cannot force to be in $H$ are $e_{p}$ and $e_{q}$. In order to
maintain a good approximation ratio we have to make sure that both are light
with respect to all objectives. This is where Lemma~\ref{lem:somelight} comes
into play: If we choose $p$ and $q$ sufficiently large, then $e_p$ and $e_q$ are light.

To fix notation, let $K$ be the of edges that we want to keep of $\tilde H$.
Let $L = L(K) = \{v \in V \mid \exists e \in K: v \in e\}$ be the set of vertices
incident to edges in $K$. Let $w^{-L}$ be defined by
\[
    w^{-L} (e)
  = \begin{cases}
    w(e) & \text{if $e \cap L = \emptyset$ and} \\
    0    & \text{if $e \cap L \neq \emptyset$.}
    \end{cases}
\]
This means that the weight of edges incident to $L$ is set to $0$, which
includes the edges in $K$.
For a bound $\beta$, $w^{-L, \beta}(e)$ is defined accordingly: $w^{-L, \beta}(e) = w^{-L}(e)$
if $w^{-L}(e) \leq \beta$ and $w^{-L, \beta}(e) = 0$ otherwise.

There are more edges of $\tilde H \setminus K$ whose weight is affected by
$w^{-L}$: Let
\[
  T = T(K) = \{e \in \tilde H \mid e \notin K, e \cap L(K) \neq \emptyset\}
\]
be the set of edges that have exactly one endpoint in $L$. The weights of these
edges are set to $0$ in $w^{-L}$, but we cannot force them to be in any cycle
cover as mentioned above. (They are the edges of $\tilde H$ that are adjacent to $K$
but not in $K$.)

The following lemma is the undirected counterpart of Lemma~\ref{lem:dirsubset}.
In particular, it takes care of the set $T$, which is only needed for the
analysis and not for the algorithm.
Let
\[
    g(k,\eps)
  =       k
    \cdot \left\lceil
             \frac{\log(1/3)}
                  {\log\bigl(1-\eta_{k,\eps/3} + (\frac \eps3)^3\bigr)}
          \right\rceil.
\]
The function $g$ plays the same role as the function $f$ in the previous subsection.

\begin{lemma}
\label{lem:undsubset}
  For every Hamiltonian cycle $\tilde H$ and every $\eps > 0$, there exists a subset
  $K \subseteq \tilde H$ of at most $g(k,\eps)$ paths, each of length at most
  $\frac{12k}{\eps} g(k, \eps)$ with the following properties: Let $L = L(K)$ and
  $T = T(K)$. For every $i \in [k]$, we have
  \begin{enumerate}
  \item $w_i(K) \geq (2/3 - \eps) \cdot w_i(\tilde H)$ or \label{uenough}
  \item $w_i^{-L}(e) \leq \bigl(\eta_{k, \eps/3} - (\frac \eps3)^3\bigr) \cdot
     w_i^{-L}(\tilde H)$ for all $e \in \tilde H$. \label{ulight}
  \end{enumerate}
  Furthermore, we have $w(T) \leq \frac{\eps}3 \cdot w(\tilde H)$.
\end{lemma}

\begin{proof}
The proof is similar to the proof of Lemma~\ref{lem:dirsubset}, but slightly
more involved since we have to keep track of the set $T$.
Let $\xi = \eta_{k, \eps/3} - (\frac \eps3)^3$.

We put paths one by one into $K$ until the properties are met.
If not all $i$ fulfill Property~\ref{enough} or~\ref{light} yet, then we have to
add another path to $K$.
In this case, there exists an edge $e_0 \in \tilde H_{-K}$ and an $i \in [k]$ such that
$w_i(e_0) > \xi \cdot w_i^{-K}(\tilde H)$ and
$w_i(K) < (\frac 23 - \eps) \cdot w_i(\tilde H)$.
We call $i$ the winner of round $j$ and add it to $K$.
Furthermore, we extend $e_0$ to both sides to obtain paths
$e_{p}, e_{p+1}, \ldots, e_0, \ldots, e_{q}$ for some $p < 0 < q$. Here, $p$ is chosen
such that either $w(e_{p}) \leq \frac{\eps}{6g(k, \eps)} \cdot w(\tilde H)$ or $e_p \in K \cup T$, and $q$ is chosen
analogously.
We can choose $-p, q \leq \frac{6kg(k, \eps)}{\eps}$.
We put $e_p$ and $e_q$ into $T$ and
$e_{p+1}, \ldots, e_{q-1}$ into $K$.

Now let $K$ be the set of edges before an iteration, and let
$K' = K \cup \{e_{p+1}, \ldots, e_{q-1}\}$ be the set of edges afterwards.
We have $w_i^{-K'} (\tilde H) + w_i(e_0) \leq w_i^{-K}(\tilde H)$.
Since $w_i(e_0) > \xi w_i^{-K}$, this yields
$w_i^{-K'}(\tilde H) < (1-\xi) \cdot w_i^{-K}(\tilde H)$.
Thus, if $i$ is the winner in $\ell$ rounds, and the resulting set of edges is $K$,
then $w_i^{-K}(\tilde H) \leq (1-\xi)^\ell \cdot w_i(\tilde H)$.
If $\ell = \lceil \log_{1-\xi} 1/3\rceil = g(k,\eps)/k$, then
\begin{equation}
  w_i^{-K}(\tilde H) \leq \frac 13 \cdot w_i(\tilde H) \: .
  \label{equ:nomuchleft}
\end{equation}
Since every round has a winner, after at most $g(k, \eps)$,
all properties are met. This is
because $w(\tilde H) = w(K) + w(T) + w^{-K}(\tilde H)$.
Any edge $e \in T$ fulfills $w(e) \leq \frac{\eps}{6 g(k, \eps)} \cdot w(\tilde H)$.
Since we put at most $2 g(k, \eps)$ edges into $T$, we have $w(T) \leq \frac{\eps}3 \cdot w(\tilde H)$.
Together with \eqref{equ:nomuchleft}, this concludes the proof.
\end{proof}

Again, given that we have the correct set $K \subseteq \tilde H$, we have to find
a cycle cover that approximates $\tilde H$ with respect to $w^{-L(K)}$.
That this can be done is shown in the following undirected counterpart of Lemma~\ref{lem:betapower}.

\begin{lemma}
\label{lem:betapowerund}
  Let $\nu > 0$ be arbitrary. Let $\tilde H$ be an undirected Hamiltonian cycle
  with $w(e) \leq \bigl(\eta_{k,\nu}  - \nu^3\bigr) \cdot w(\tilde H)$ for all
  $e \in \tilde H$. Let $\beta = \beta(\tilde H)$, and let $\pc$ be a
  $(1-\nu)$ approximate Pareto curve of cycle covers with respect to~$w^\beta$.

  Then $\pc$ contains a cycle cover $C$ with
  $w^\beta(C) \geq (1-\nu) \cdot w(\tilde H)$ and
  $w^\beta(e) \leq \eta_{k,\nu} \cdot w^\beta(C)$ for all $e \in C$. This
  cycle cover $C$ yields a decomposition $P \subseteq C$ with
  $w(P) \geq (\frac 23 - 2\nu) \cdot w(\tilde H)$.
\end{lemma}

\begin{proof}
The proof is almost identical to the proof of
Lemma~\ref{lem:betapower} and thus omitted.
\end{proof}

Now we have everything for algorithm \spa\
(Algorithm~\ref{alg:maxstsp}) and Theorem~\ref{thm:mcstsp}.

\begin{algorithm}[t]
\begin{algorithmic}[1]
\item[] $\ptsp \leftarrow \stspalg(G, w, \eps)$
\Input undirected complete graph $G=(V,E)$, $w: E \to \ratio_+^k$, $\eps > 0$
\Output $(\frac 23 - \eps)$ approximate Pareto curve \ptsp\
   for \kcmaxstsp k with probability at least $\frac 12$
\ForAll{$K \subseteq E$ that consist of at most $g(k, \eps)$ paths of
     length at most $\leq 12kg(k, \eps)/\eps$}
\ForAll{bounds $\beta$}
   \State $L \leftarrow L(K)$
      \State $\pc_{L,\beta} \leftarrow
           \algcc(G,w^{-L,\beta}, \frac{\eps}3,
                       \frac{1}{2n^{2k+12k g^2(k, \eps)/\eps}})$
           \label{line:cccund}
      \ForAll{$I \subseteq [k]$ and $C \in \pc_{L,\beta}$}
         \If{$w_I^{-L,\beta}(e) \leq \eta_{k,\eps/3}
              \cdot w_I^{-L,\beta}(C)$ for all $e \in C$}
         \State $P \leftarrow
              \lightweight(C, w^{-L,\beta}_I, \frac\eps4)$
         \State remove edges of weight $0$ from $P$ to get $P'$
         \State add edges to $K \cup P'$ to get a Hamiltonian cycle $H$; add $H$ to $\ptsp$
      \EndIf
   \EndFor
\EndFor
\EndFor
\end{algorithmic}
\caption{Approximation algorithm for \kcmaxstsp k.}
\label{alg:maxstsp}
\end{algorithm}

\begin{theorem}
  \label{thm:mcstsp}
  For every $k \geq 2$ and $\eps > 0$, \spa\ is a randomized
  $\frac 23 - \eps$ approximation for $k$-criteria \maxstsp\ whose running-time
  is polynomial in the input size.
\end{theorem}

\begin{proof}
  Let us first concentrate on the approximation ratio. Consider any Hamiltonian
  cycle $\tilde H$. We have to show that $\ptsp$ contains a Hamiltonian cycle
  $H$ with $w(H) \geq (2/3 - \eps) w(\tilde H)$. According to Lemma~\ref{lem:undsubset},
  there exists a set $K \subseteq \tilde H$ and a subset $I \subseteq [k]$ of
  the objectives such that the following holds:
  \begin{itemize}
  \item For every $i \in [k] \setminus I$, we have
     $w_i(K) \geq (2/3 - \eps) \cdot w_i(\tilde H)$.
  \item For every $i \in I$ and for every edge $e \in H_{-K}$, we have
     $w_i(e) \leq \bigl(\eta_{k, \eps/3} - (\frac \eps 3)^3\bigr) \cdot w_i(H_{-K})$.
  \end{itemize}

  Let $L = L(K)$, and let $\beta = \beta(\tilde H)$ with respect to the edge weights
  $w^{-L}$. According to Lemma~\ref{lem:betapowerund} with edge weights $w^{-L}$ and $\nu = \eps/3$,
  the set $\pc_{L, \beta}$
  contains a cycle cover $C$ from which we get a decomposition $P \subseteq C$ with
  $w_i^{-L}(P') = w_i^{-L}(P) \geq (\frac 23 - \frac{2\eps}3) \cdot w_i^{-L}(\tilde H)$. ($P'$
  is obtained from $P$ by removing edges of weight $0$.)
  Since $w(\tilde H) = w^{-L}(\tilde H) + w(T) + w(K)$ and $w(T) \leq \frac{\eps}3 \cdot w(\tilde H)$
  according to Lemma~\ref{lem:undsubset}, we have
  $w(H) \geq w(P' \cup K) \geq (2/3 - \eps) \cdot w(\tilde H)$, which is enough.

  The error probabilities of the randomized computations in
  line~\ref{line:cccund} sum up to at most $1/2$ since there are at most
  $n^{2k}$ bounds $\beta$ and at most $n^{6kg^2(k, \eps)}$ subsets $K$. By a
  union bound, the probability that one of the computations fails is thus at most
  $1/2$. The running time follows since $g(k, \eps)$ is independent of $n$,
  the number of bounds $\beta$ is bounded by $n^{2k}$, and the number of $I$ is
  $2^k$.
\end{proof}

%%%%%%%%%%%%%%%%%%%%%%%%%%%%%%%%%%%%%%%%%%%%%%%%%%%%%%%
\section{Deterministic Approximations for \kcmaxstsp 2}
\label{sec:detbi}
%%%%%%%%%%%%%%%%%%%%%%%%%%%%%%%%%%%%%%%%%%%%%%%%%%%%%%%

The algorithms presented in the previous section are randomized due to the
computation of approximate Pareto curves of cycle covers. So are most
approximation algorithms for multi-criteria TSP with the exception
of a simple $(2+\eps)$ approximation for
\kcminstsp k~\cite{MantheyRam:MultiCritTSP:2009}.

As a first step towards deterministic approximation algorithms for multi-criteria
Max-TSP,
we present a deterministic $7/27 \approx 0.26$ approximation for
\kcmaxstsp 2. The key insight for the results of this section is the following
lemma.

\begin{lemma}
\label{lem:hammatch}
  Let $M$ be a (not necessarily perfect) matching, let $H$ be a collection of paths or a Hamiltonian
  cycle, and let $w$ be edge weights ($w$ is a single-criterion function).
  Then there exists a subset $P \subseteq H$
  such that
  \begin{enumerate}[(i)]
  \item $P \cup M$ is a collection of paths or a Hamiltonian cycle (we call $P$
        in this case an $M$-feasible set) and
  \item $w(P) \geq w(H)/3$.
  \end{enumerate}
\end{lemma}

\begin{proof}
  We prove the lemma by induction on $|M| + |H|$. For $|M| + |H| = 0$, the
  lemma is trivially true.
  As induction hypothesis, we assume the lemma holds for all $M$ and $H$ with
  $|M| + |H| < \ell$, and we want to prove it for $|M| + |H| = \ell$.

  We distinguish two cases. The first case is $M \cap H \neq \emptyset$. Then
  we set $\tilde P = M \cap H$ and $H' = H \setminus M$.
  By induction
  hypothesis, there exists an $M$-feasible $P' \subseteq H'$ such that $w(P') \geq w(H')/3$.
  Since $\tilde P \subseteq M$, also $P = P' \cup \tilde P$ is
  $M$-feasible. Observing that
  $w(P) = w(P') + w(\tilde P) \geq w(H \cap M) + w(H \setminus M)/3 \geq w(H)/3$
  completes this case.
  
  The second case is that $M$ and $H$ are disjoint. Let
  $e = \argmax\{w(e) \mid e \in H\}$ be a heaviest edge of $H$, and let
  $f_1, f_2 \in H$ be the two edges of $H$ that are incident to $e$. Let
  $H' = H \setminus \{e, f_1, f_2\}$. (It can happen that $f_1$ or $f_2$ do not
  exist, namely if $H$ is not a Hamiltonian cycle but a collection of paths. But
  this is fine.)

  Let us first treat the case that $e$ is incident to two edges $z_1, z_2 \in M$
  of the matching.
  Then we contract $z_1$ and $z_2$ to a single edge $z$ that connects the two
  endpoints of $z_1$ and $z_2$ that are not incident to $e$ and remove the two
  vertices incident to $e$ (see Figure~\ref{fig:contraction}). Let
  $M' = (M \setminus \{z_1, z_1\}) \cup \{z\}$. Since $e, f_1, f_2$ are removed,
  $H'$ and $M'$ are a valid instance for the lemma, i.e., $M'$ is a matching and
  $H'$ is a collection of paths ($H'$ cannot be a Hamiltonian cycle). We can
  apply the induction hypothesis since $|M'| + |H'| < \ell$.

  In this way, we obtain an $M'$-feasible set $P' \subseteq H'$ with
  $w(P') \geq w(H')/3$. Set $P = P' \cup \{e\}$. Since
  $w(e) \geq w(\{e, f_1, f_2\})/3$, we have $w(P) \geq w(H)/3$. Since $P'$
  is $M'$-feasible, the set $P$ is $M$-feasible by construction.

  What remains to be considered is the case the $e$ is not incident to two edges
  $z_1, z_2 \in M$. Then we consider the shortest path in
  $e_1, \ldots, e_q \in H$ of edges in $H$ that includes $e$ such that $e_1$ and
  $e_q$ are incident to any edges $z_1, z_2 \in M$. The reasoning above holds in
  exactly the same way if replace $e$ by the path $e_1, \ldots, e_q$, and we put
  $e_1, \ldots, e_q$ into $P$. If no such path exists, then either
  $M = \emptyset$, which is easy to handle, or the path containing $e$
  ends somewhere at a vertex of degree $1$ in $H \cup M$. In the latter case, we
  can simply put the whole path into~$P$.
\end{proof}

\begin{figure}[t]
\centering
\subfigure[]{\label{sfig:c1} \includegraphics{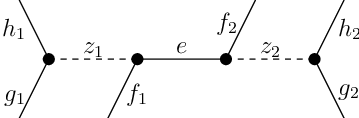}}
\qquad\qquad
\subfigure[]{\label{sfig:c2} \includegraphics{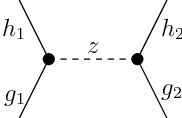}}
\caption{Contraction for the proof of Lemma~\ref{lem:hammatch}:
         We keep $e$ and
         remove $f_1, f_2$ \eqref{sfig:c1}. Then we can contract $z_1$ and $z_2$ to $z$
         \eqref{sfig:c2}.}
\label{fig:contraction}
\end{figure}

Lemma~\ref{lem:hammatch} yields tight bounds for the existence of approximate
Pareto curves with only a single element. This is the purpose of the following
theorem.

\begin{theorem}
\label{thm:twoham}
\begin{enumerate}[(a)]
\item \label{part1} For every undirected complete graph $G$ with edge weights
      $w_1$ and $w_2$, there exists a Hamiltonian cycle $H$ such that $\{H\}$ is
      a $1/3$ approximate Pareto curve for \kcmaxstsp 2.

\item \label{part2} Part~\eqref{part1} is tight: There exists a graph $G$ with
      edge weights $w_1$ and $w_2$ such that, for all $\eps > 0$, no single
      Hamiltonian tour of $G$ is a $(1/3 + \eps)$ approximate Pareto curve.
\end{enumerate}
\end{theorem}

Before embarking on the proof of the theorem, let us remark
that single-element approximate Pareto curves exist for no other
variant of multi-criteria TSP than \kcmaxstsp 2: For \kcmaxstsp k for
$k \geq 3$, we can consider a vertex incident to three edges of weight
$(1,0,0)$, $(0,1,0)$, and $(0,0,1)$, respectively. All other edges of the graph
have weight $0$. Then no single Hamiltonian cycle can have positive weight with
respect to all three objectives simultaneously. Similarly, no such result is
possible for \kcmaxatsp k and for \kcmintsp k for any $k \geq 2$.

\begin{proof}
  Let $H_1$ and $H_2$ be Hamiltonian cycles of $G$ such that $H_1$ maximizes $w_1$
  and $H_2$ maximizes $w_2$. Then there exists a matching $M \subseteq H_1$ with
  $w(M) \geq w(H_1)/3$. (We can actually get $w(H_1)/2$ if $G$ has an even
  number of vertices and $\frac{n-1}{2n} \cdot w(H_1)$ if the number $n$ of
  $G$'s vertices is odd. This, however, does not improve the result.) We apply
  Lemma~\ref{lem:hammatch} with $H = H_2$ and obtain an $M$-feasible set
  $P \subseteq H_2$. From $M$ and $P$, we obtain a Hamiltonian cycle
  $H' \supseteq M \cup P$: Either $M \cup P$ is already a Hamiltonian cycle,
  then nothing has to be done. Or $M \cup P$ is a collection of paths. Then we
  add appropriate edges to obtain $H'$. We claim that $\{H'\}$ is a $1/3$
  approximate Pareto curve: Let $\tilde H$ be any Hamiltonian tour. Then
  \[
    w_1(H') \geq w_1(M) \geq w_1(H_1)/3 \geq w_1(\tilde H)/3
  \]
  and
  \[
    w_2(H') \geq w_2(P) \geq w_2(H_2)/3 \geq w_2(\tilde H)/3.
  \]

  To finish the proof, let us show that ratio $1/3$ is tight.
  Consider the graph in Figure~\ref{fig:pentagon}.
  The solid edges plus two dotted edges form a Hamiltonian cycle of weight
  $(3,0)$. The dashed edges plus two other dotted edges form a Hamiltonian cycle of
  weight $(0,3)$. To get a single-element $(\frac 13 + \eps)$-approximate Pareto curve $\{H\}$, we must
  have $w_i(H) \geq 1 + 3\eps$ for $i \in \{1,2\}$. Thus, the Hamiltonian cycle $H$ must contain two solid edges and two
  dashed edges, which is impossible.
\end{proof}

\begin{figure}[t]
\centering
\includegraphics{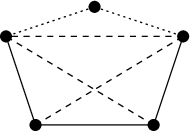}
\caption{The graph for Theorem~\ref{thm:twoham}\eqref{part2}. Solid edges
are of weight $(1,0)$, dashed edges are of weight $(0,1)$, dashed edges
have weight $(0,0)$.}
\label{fig:pentagon}
\end{figure}

\begin{algorithm}[t]
\begin{algorithmic}[1]
\item[] $\ptsp \leftarrow \twoalg(G, w_1, w_2)$
\Input undirected complete graph $G=(V,E)$, edge weights
       $w_1, w_2:  E \to \ratio_+^k$
\Output a $7/27$ approximate Pareto curve $H$ for \kcmaxstsp 2
\State compute a maximum-weight matching $M$ with respect to $w_1$
\State compute a $7/9$ approximate Hamiltonian cycle $H_2$ with respect
       to $w_2$
\State $P \leftarrow H_2 \cap M$
\State $M' \leftarrow M$
\State $H_2 \leftarrow H_2 \setminus P$
\While{$H_2 \neq \emptyset$}
   \State $e \leftarrow \argmax\{w_2(e') \mid e' \in H_2\}$
   \State extend $e$ to a path $e_1, \ldots, e_q \in H_2$ such that only $e_1$
          and $e_q$ are incident to edges $z_1, z_2 \in M'$ or the path cannot be
          extended anymore
   \State $P \leftarrow P \cup \{e_1, \ldots, e_q\}$
   \State $H_2 \leftarrow H_2 \setminus \{e_1, \ldots, e_q\}$
   \If{$z_1$ or $z_2$ exists}
       \State let $f_1, f_2 \in H_2$ be the two edges extending the path if they
              exist
       \State $H_2 \leftarrow H_2 \setminus \{f_1, f_2\}$
       \If{both $z_1$ and $z_2$ exist}
          \State contract $z_1$ and $z_2$ to $z$
          \State $M' \leftarrow (M' \setminus \{z_1, z_2\}) \cup \{z\}$
   \EndIf
   \EndIf
\EndWhile
\State let $H$ be a Hamiltonian cycle obtained from $P \cup M$
\end{algorithmic}
\caption{Approximation algorithm for \kcmaxstsp 2.}
\label{alg:maxtwo}
\end{algorithm}

Lemma~\ref{lem:hammatch} and Theorem~\ref{thm:twoham} are constructive in the
sense that, given a Hamiltonian cycle $H_2$ that maximizes $w_2$, the tour $H$
can be computed in polynomial time. A matching $M$ with
$w_1(M) \geq w_1(H_1)/3$ can be computed in cubic time. However, since we cannot
compute an optimal $H_2$ efficiently, the results cannot be exploited directly to get an algorithm. Instead, we
use an approximation algorithm for finding a Hamiltonian tour with as much
weight with respect to $w_2$ as possible. Using the $7/9$ approximation
algorithm for \maxstsp~\cite{PaluchEA:MaxTSP79:2009}, we obtain
Algorithm~\ref{alg:maxtwo} (which is an algorithmic version of
Lemma~\ref{lem:hammatch}) and the following theorem.

\begin{theorem}
  \mt\ is a deterministic $7/27$ approximation algorithm with running-time
  $O(n^3)$ for \kcmaxstsp 2.
\end{theorem}

\begin{proof}
The running-time is dominated by the running-time of the $7/9$ approximation
for \maxstsp\ by Paluch et al.~\cite{PaluchEA:MaxTSP79:2009} and the time for
computing the matching, both of which is $O(n^3)$. The approximation ratio
follows from $\frac{7}{9} \cdot \frac 13 = \frac{7}{27}$.
\end{proof}

%%%%%%%%%%%%%%%%%%%%%%%%%%%%%%%%%%%%%%%%%%%%%%%%%%%%%%%%%%%%%
\section{Approximation Algorithm for Multi-Criteria Min-ATSP}
\label{sec:atsp}
%%%%%%%%%%%%%%%%%%%%%%%%%%%%%%%%%%%%%%%%%%%%%%%%%%%%%%%%%%%%%

Now we turn to \kcminatsp k, i.e., Hamiltonian cycles of \emph{minimum}
weight are sought in directed graphs. Algorithm~\ref{alg:minatsp} is an
adaptation of the algorithm of Frieze et al.~\cite{FriezeEA:TSP:1982} to
multi-criteria ATSP. Therefore, we briefly describe their algorithm: We compute
a cycle cover of minimum weight. If this cycle cover is already a Hamiltonian
cycle, then we are done. Otherwise, we choose an arbitrary vertex from every
cycle. Then we proceed recursively on the subset of vertices thus chosen to
obtain a Hamiltonian cycle that contains all these vertices. The cycle cover
plus this Hamiltonian cycle form an Eulerian graph. We traverse the Eulerian
cycle and take shortcuts whenever visiting vertices more than once. The
approximation ratio achieved by this algorithm is $\log_2 n$ for
\minatsp~\cite{FriezeEA:TSP:1982}.

\begin{algorithm}[t]
\begin{algorithmic}[1]
\item[] $\ptsp \leftarrow \minalg(G, w, \eps)$
\Input directed complete graph $G=(V,E)$, $n = |V|$, edge weights
       $w: E \to \ratio_+^k$, $\eps > 0$
\Output $(\log n + \eps)$ approximate Pareto curve for \kcminatsp k
        with probability at least $1/2$
\State $\eps' \leftarrow \eps^2/\log^3 n$
\State $\pp_0 \leftarrow \{\emptyset\}$
\For{$j \leftarrow 1 \text{ to } \lfloor \log_2 n \rfloor$}
       \label{line:while}
  \State $\pp_{j} \leftarrow \emptyset$
         \label{line:jcomp}
  \ForAll{$C \in \pp_{j-1}$}
    \If{$(V, C)$ is connected}
      \State add $C$ to $\pp_{j}$
             \label{line:pffill}
    \Else
      \State select one vertex of every component of $(V, C)$ to obtain
             $V'$
             \label{line:pareto}
      \State $\pc \leftarrow \mincc(V', w, \eps', \frac 1{2Q \log n})$
             \label{line:cccomp}
\Comment $Q$ is defined in Lemma~\ref{lem:running}
      \State $\pp_{j} \leftarrow \pp_{j} \cup \{C \cup C' \mid C' \in \pc\}$
             \label{line:paretoend}
    \EndIf
  \EndFor
  \While{there are $C, C' \in \pp_{j}$ with the same $\eps'$-signature}
         \label{line:sparse}
    \State remove one of them arbitrarily
  \EndWhile\label{line:sparseend}
  \State $j \leftarrow j+1$
\EndFor \label{line:jcompend}
\State $\ptsp \leftarrow \emptyset$
       \label{line:hamconstruct}
\ForAll{$C \in \pp_{\lfloor \log_2 n \rfloor}$}
  \State walk along the Eulerian cycle of $C$, take shortcuts
  to obtain a Hamiltonian cycle $H$
  \State add $H$ to $\ptsp$
\EndFor \label{line:hamend}
\end{algorithmic}
\caption{Approximation algorithm for \kcminatsp k.}
\label{alg:minatsp}
\end{algorithm}

To approximate \kcminatsp k, we use \mia\ (Algorithm~\ref{alg:minatsp}), which
proceeds as follows: We compute an approximate Pareto curve of cycle covers of minimum weight.
This is done by \mca, where $\mincc(G, w, \eps, p)$ computes a $(1+\eps)$
approximate Pareto curve of cycle covers of $G$ with weights $w$ with a success
probability of at least $1-p$ in time polynomial in the input size, $1/\eps$,
and $\log(1/p)$.
Then we iterate by computing approximate Pareto curves of cycle covers
on vertex sets $V'$ for every cycle cover $C$ in the previous set. The set $V'$
contains exactly one vertex of every cycle of $C$. Unfortunately, it can happen
that we construct a super-polynomial number of solutions in this way. To cope
with this, we remove some intermediate solutions if there are other intermediate
solutions whose weight is close by. We call this process \emph{sparsification}.
It is performed in lines~\ref{line:sparse} and~\ref{line:sparseend} of
Algorithm~\ref{alg:minatsp} and based on the following observation: Let
$\eps > 0$, and consider a set $H$ of edges of weight $w(H) \in \ratio_+^k$. For
every $i \in \{1,\ldots, k\}$ with $w_i(H) \neq 0$, there is a unique
$\ell_i \in \nat$ such that
$w_i(H) \in \bigl[(1+\eps)^{\ell_i}, (1+\eps)^{\ell_i+1}\bigr)$. If
$w_i(H) = 0$, then we let $\ell_i = -\infty$. We call the
vector $\ell = (\ell_1, \ldots, \ell_k)$ the \emph{$\eps$-signature} of $H$ and
of $w(H)$. Since $w(H) \in \bigl([2^{-p(N)}, 2^{p(N)}] \cup \{0\}\bigr)^k$,
where $N$ is the size of the instance,
the number of possible values of $\ell_i$ is bounded by a polynomial
$q(N, 1/\eps)$. There are at most $q^k$
different $\eps$-signatures, which is polynomial for fixed $k$. To get an
approximate Pareto curve, we can restrict ourselves to at most one solution
for any $\eps$-signature.

The set $\pp_0$ is initialized with the empty set of edges.
In the loop in lines~\ref{line:while}
to~\ref{line:jcompend}, the algorithm computes iteratively Pareto curves of
cycle covers. The set $\pp_j$ contains sets $C$ of edges consisting of cycle covers:
Given a $C \in \pp_{j-1}$, $\pp_j$ contains cycle covers
on the graph consisting of one node for every connected component of $C$.
If $(V,C)$ is already connected, then $C$ is simply put into $\pp_j$ without modification.
In lines~\ref{line:sparse} and~\ref{line:sparseend}, the sparsification takes
place. Finally, in lines~\ref{line:hamconstruct} to~\ref{line:hamend},
Hamiltonian cycles are constructed from the Eulerian graphs.

Let us now come to the analysis of the algorithm. Our goal is to prove the
following result, which follows from Lemmas~\ref{lem:running}, \ref{lem:ratio},
and~\ref{lem:prob12} below. \mia\ is the first approximation algorithm for
\kcminatsp k.

\begin{theorem}
\label{thm:minatsp}
  For every $\eps > 0$ and $k \geq 2$, Algorithm~\ref{alg:minatsp} is a randomized
  $(\log n + \eps)$ approximation for \kcminatsp k\
  with a success probability of at least $1/2$. Its
  running-time is polynomial in the input size and~$1/\eps$.
\end{theorem}

We observe that for every $j \in \{0,1,\ldots, \lfloor \log_2 n \rfloor\}$ and $C \in \pp_j$, the graph
$(V, C)$ consists of at most $n/2^j$ connected components. For $j = 0$, this holds
since $(V, \emptyset)$ consists of $n$ connected components.
For $j > 0$ and
$C \in \pp_{j-1}$, $(V,C)$ consists of at most $n/2^{j-1}$ connected components
by the induction hypothesis. If $(V,C)$ is connected, then
$C \in \pp_{j}$, and the claim holds since $n/2^j \geq n/2^{\lfloor \log_2 n \rfloor} \geq 1$.
Otherwise, since every cycle involves at least two vertices,
the claim holds also for $\pp_j$.
This yields that $(V,C)$ is connected for all $C \in \pp_{\lfloor \log_2 n\rfloor}$:
Such a $(V,C)$ consists of at most $n/2^j \leq n/2^{\lfloor \log_2 n\rfloor} < 2$
connected components.

Let us now analyze the running-time. After that, we examine the approximation
performance and finally the success probability.

\begin{lemma}
\label{lem:running}
  The running-time of Algorithm~\ref{alg:minatsp} is polynomial in the input
  size and $1/\eps$.
\end{lemma}

\begin{proof}
  Let $N$ be the size of the instance at hand, and let $Q=Q(N, 1/\eps')$ be a
  two-variable polynomial that bounds the number of different $\eps'$-signatures
  of solutions for instances of size at most $N$. We abbreviate ``polynomial in
  the input size and $1/\eps$'' simply by ``polynomial.'' This is equivalent to
  ``polynomial in the input size and $1/\eps'$'' by the choice of $\eps'$.

  The approximate Pareto curves can be computed in polynomial time with a
  success probability of at least $1-(2Q \log n)^{-1}$ by executing the
  randomized FPTAS $\lceil \log(2Q \log n) \rceil$ times. Thus, all operations
  can be implemented to run in polynomial time provided that the cardinalities
  of all sets $\pp_j$ are bounded from above by a polynomial $Q$ for all $j$.
  Then, for each $j$, at most $Q$ approximate Pareto curves of cycle covers are
  constructed in line~\ref{line:cccomp}.

  For every $\eps'$-signature and every $j$, the set $\pp_j$ contains at most
  one set of edges for any $\eps'$-signature. The lemma follows since
  the number of different $\eps'$-signatures is bounded by~$Q$.
\end{proof}

Let us now analyze the approximation ratio. To do so, we will assume that all
randomized computations of $(1+\eps')$ approximate cycle covers are successful.

\begin{lemma}
\label{lem:ratio}
  Assume that in all executions of line~\ref{line:cccomp} of
  Algorithm~\ref{alg:minatsp} an $(1+\eps')$ approximate Pareto curve of
  cycle covers is successfully computed. Then Algorithm~\ref{alg:minatsp}
  achieves an approximation ratio of
  $\log_2 n + \eps$ for \kcminatsp k.
\end{lemma}

\begin{proof}
  Let $\tilde H$
  be any Hamiltonian cycle on $V$. We have to show that the set $\ptsp$ of
  solutions computed by Algorithm~\ref{alg:minatsp} contains a Hamiltonian cycle
  $H$ with $w(H) \leq (\log_2 n + \eps) \cdot w(\tilde H)$.
  
  Given any $C \in \pc_{\lfloor \log_2 n \rfloor}$, due to the triangle inequality,
  we construct a Hamiltonian cycle $H$ in lines~\ref{line:hamconstruct}
  to~\ref{line:hamend} such that $w(H) \leq w(C)$.
  What remains to be proved is that, for every Hamiltonian cycle $\tilde H$,
  there exists a $C \in \pc_{\lfloor \log_2 n \rfloor}$ such that
  $w(C) \leq (\log n+ \eps) \cdot w(\tilde H)$.

\begin{lemma}
\label{cla:induction}
  For every $j$, there exists a $C \in \pp_{j}$
  with $w(C) \leq (1+\eps')^{j} \cdot j \cdot w(\tilde H)$.
\end{lemma}

\begin{proof}
  The proof is by induction on $j$. For $j = 0$, the lemma holds since $w(\emptyset) = 0$.

  Now assume that the lemma holds for $j -1$ for $j > 0$.
  Consider any $C \in \pp_{j-1}$ that satisfies the lemma for $j-1$ and $\tilde H$.
  Such a $C$ exists by the induction hypothesis.
  If $(V,C)$ is connected, then $C \in \pp_j$, and $C$ satisfies the lemma also for $j$.
  Otherwise, let $V'$ be the set of vertices chosen from the connected components of $(V,C)$ in
  line~\ref{line:pareto}.
  Let $\tilde H'$ be $\tilde H$ restricted to $V'$ by taking shortcuts. By
  the triangle inequality, we have $w(\tilde H') \leq w(\tilde H)$. After
  line~\ref{line:cccomp}, $\pc$ contains a cycle cover $C'$ with
  $w(C') \leq (1+\eps') \cdot w(\tilde H')$. Then
  \[
     w(C' \cup C) \leq \bigl((1+\eps')^{j-1} \cdot (j-1)  +
                   (1+\eps')\bigr)
          \cdot w(\tilde H).
  \]
  What remains to be analyzed is the sparsification in lines~\ref{line:sparse}
  to~\ref{line:sparseend}. After that 
  $\pp_j$ contains a $C''$ (with might coincide with $C \cup C'$)
  with with the same
  $\eps'$-signature as $C \cup C'$. Thus,
  \begin{align*}
         w(C'') & \leq (1+\eps') \cdot w(C \cup C'') \leq (1+\eps') \cdot
         \bigl((1+\eps')^{j} \cdot j +
               (1+\eps')\bigr) \cdot w(\tilde H) \\
   &  \leq (1+\eps')^{j+1} \cdot (j+1) \cdot w(\tilde H),
  \end{align*}
  and $C''$ fulfills the requirements of the lemma.
\end{proof}

  Since every $C \in \pc_{\lfloor \log_2 n\rfloor}$ yields a Hamiltonian cycle
  without increasing the weight, we obtain an approximation
  ratio of
  \[
          \log_2 n \cdot (1+\eps')^{\lfloor \log n \rfloor}
      \leq  \log_2 n  \cdot \left(1+\frac{\eps^2}{\log^3 n}\right)^{\log n}
     \leq  \log_2 n  \cdot \exp\left(\frac{\eps^2}{\log^2 n}\right)
     \leq  \log_2 n  + \eps.
  \]
  The first inequality follows from our choice of $\eps'$. The second inequality
  holds since $(1+\frac xy)^y \leq \exp(x)$. The third inequality
  holds because $\exp(x^2) \leq 1+x$ for $x \in [0, 0.7]$ (we assume
  $\eps/\log n < 0.7$ without loss of generality).
\end{proof}

\begin{lemma}
\label{lem:prob12}
  The probability that in a run of Algorithm~\ref{alg:minatsp}
  every execution of
  line~\ref{line:cccomp} is successful is at least $1/2$.
\end{lemma}

\begin{proof}
  Line~\ref{line:cccomp} of
  Algorithm~\ref{alg:minatsp} is executed at most $Q \cdot \log n$ times, where
  $Q$ is an upper bound for the number of different $\eps'$-signatures of
  solutions of instances of size at most $N$.
  Each execution fails with a probability of at most
  $\frac{1}{2Q \log n}$. Thus by a union bound, the probability that one of
  them fails is at most $1/2$.
\end{proof}

Since randomization is only needed for \mca, Lemma~\ref{lem:prob12} completes
the proof of Theorem~\ref{thm:minatsp}.

\begin{remark}
According to Bl\"aser et al.~\cite{BlaeserEA:ATSP:2006}, the algorithm of Frieze et al.~\cite{FriezeEA:TSP:1982}
can be turned into a $\frac 1{1-\gamma}$-approximation for \minatsp\ with $\gamma$-triangle inequality
for $\gamma \in [\frac 12, 1)$. An instance fulfills the $\gamma$-triangle inequality,
if $w(u,v) \leq \gamma \cdot (w(u,x) + w(x,v))$ for all distinct $u,v,x$.
In the same way, Algorithm~\ref{alg:minatsp} can be turned into
a $\frac{1}{1-\gamma} + \eps$ approximation
for this variant of multi-criteria TSP. This improves over existing results~\cite{MantheyRam:MultiCritTSP:2009}
for $\gamma \geq 0.55$.
\end{remark}

%%%%%%%%%%%%%%%%%%%%%%%%%%%%
\section{Concluding Remarks}
\label{sec:remarks}
%%%%%%%%%%%%%%%%%%%%%%%%%%%%

We have presented approximation algorithms for almost all variants of
multi-criteria TSP. The approximation ratios of our algorithms are independent
of the number $k$ of criteria and come close to the currently best ratios for
TSP with a single objective. Furthermore, they work for any
number of criteria.

Furthermore, we have devised a
deterministic $7/27$ approximation for \kcmaxstsp 2\
with cubic running-time, and we proved
that for \kcmaxstsp 2, there always exists a $1/3$ approximate Pareto curve that
consists of a single element.

Most approximation algorithms for multi-criteria TSP use randomness
since computing approximate Pareto curves of cycle covers requires randomness.
This raises the question of whether there are algorithms
for multi-criteria TSP that are faster, deterministic, and achieve better
approximation ratios.

%%%%%%%%%%%%%%%%%%%%%%%%%%
\section*{Acknowledgement}
%%%%%%%%%%%%%%%%%%%%%%%%%%

I thank Markus Bl\"aser and Mahmoud Fouz for fruitful discussions as well
as the anonymous referees for their helpful comments.

%%%%%%%%%%%%%%%%%%%%%%%%%%%%%%%%%%%%%%%%%%%%%%%%%%%%%%%%%%%%%%%%%%%%%%%%%%%%%%%%

%\bibliographystyle{plain}
%\bibliography{abbrev,papers,chapters,books,bodo}

\end{document}